\documentclass[envcountsame]{llncs}
\usepackage[latin1]{inputenc}
\usepackage{longtable,helvet,courier,fancyhdr}
\usepackage{amsfonts}
\usepackage{amsmath,amssymb,bbm}
\usepackage{multicol}
\usepackage[graph]{xy}

\def\bq#1#2{\beta_{#1}^{(#2)}}
\def\cocoa{{\hbox{\rm C\kern-.13em o\kern-.07em C\kern-.13em o\kern-.15em A}}}
\def\F{\mathcal{F}}
\def\G{\mathcal{G}}
\def\H{\mathcal{H}}
\def\I{\mathcal{I}}
\def\Isat{\mathcal{I}^{\mathrm{sat}}}
\def\J{\mathcal{J}}
\def\kk{\mathbbm{k}}
\def\M{\mathcal{M}}
\def\mf{\mathfrak{m}}
\def\mult#1#2#3{{#1}_{#2}(#3)}
\def\NN{\mathbbm{N}}
\def\nmult#1#2#3{\overline{#1}_{#2}(#3)}
\def\P{\mathcal{P}}
\def\V{\mathcal{V}}
\def\X{\mathcal{X}}

\DeclareMathOperator{\cls}{cls}
\DeclareMathOperator{\ch}{char}
\DeclareMathOperator{\depth}{depth}
\DeclareMathOperator{\gin}{gin}
\DeclareMathOperator{\lt}{lt}
\DeclareMathOperator{\pd}{pd}
\DeclareMathOperator{\reg}{reg}
\DeclareMathOperator{\sat}{sat}

\begin{document}
\title{Quasi-Stability versus Genericity}
\author{Amir Hashemi\inst{1} \and Michael Schweinfurter\inst{2} 
        \and Werner M. Seiler\inst{2}}
\institute{Department of Mathematical Sciences, 
Isfahan University of Technology\\ Isfahan, 84156-83111, Iran;
\email{Amir.Hashemi@cc.iut.ac.ir}
\and Institut f\"{u}r Mathematik,
Universit\"at Kassel\\
Heinrich-Plett-Stra\ss e 40, 34132 Kassel, Germany\\
\email{[michael.schweinfurter,seiler]@mathematik.uni-kassel.de}
}
\maketitle

\begin{abstract}
  Quasi-stable ideals appear as leading ideals in the theory of Pommaret
  bases.  We show that quasi-stable leading ideals share many of the
  properties of the generic initial ideal.  In contrast to genericity,
  quasi-stability is a characteristic independent property that can be
  effectively verified.  We also relate Pommaret bases to some invariants
  associated with local cohomology, exhibit the existence of linear quotients
  in Pommaret bases and prove some results on componentwise linear ideals.
\end{abstract}

\section{Introduction}

The generic initial ideal of a polynomial ideal
$0\neq\I\unlhd\P=\kk[\X]=\kk[x_{1},\dots,x_{n}]$ was defined by Galligo
\cite{gall:weier} for the reverse lexicographic order and $\ch{\kk}=0$; the
extension to arbitrary term orders and characteristics is due to Bayer and
Stillman \cite{bs:reverse}.  Extensive discussions can be found in
\cite[Sect.~15.9]{de:ca}, \cite[Chapt.~4]{hh:monid} and \cite{mlg:gin}.  A
characteristic feature of the generic initial ideal is that it is Borel-fixed,
a property depending on the characteristics of $\kk$.

Quasi-stable ideals are known under many different names like ideals of nested
type \cite{bg:scmr}, ideals of Borel type \cite{hpv:ext} or weakly stable
ideals \cite{cs:reg}.  They appear naturally as leading ideals in the theory
of Pommaret bases \cite{wms:comb2}, a special class of Gr\"obner bases with
additional combinatorial properties.  The notion of quasi-stability is
characteristic independent.

The generic initial ideal has found quite some interest, as many invariants
take the same value for $\I$ and $\gin{\I}$, whereas arbitrary leading ideals
generally lead to larger values.  However, there are several problems with
$\gin{\I}$: it depends on $\ch{\kk}$; there is no effective test known to
decide whether a given leading ideal is $\gin{\I}$ and thus one must rely
on expensive random transformations for its construction.  The main point of
the present work is to show that quasi-stable leading ideals enjoy many of the
properties of $\gin{\I}$ and can nevertheless be effectively detected and
deterministically constructed.

Throughout this article, $\P=\kk[\X]$ denotes a polynomial ring in the
variables $\X=\{x_{1},\dots,x_{n}\}$ over an infinite field $\kk$ of arbitrary
characteristic and $0\neq\I\lhd\P$ a proper homogeneous ideal.  When
considering bases of $\I$, we will always assume that these are homogeneous,
too.  $\mf=\langle\X\rangle\lhd\P$ is the homogeneous maximal ideal.  In order
to be consistent with \cite{wms:comb1,wms:comb2}, we will use a non-standard
convention for the reverse lexicographic order: given two arbitrary terms
$x^{\mu},x^{\nu}$ of the same degree, $x^{\mu}\prec_{\mathrm{revlex}}x^{\nu}$
if the first non-vanishing entry of $\mu-\nu$ is positive.  Compared with the
usual convention, this corresponds to a reversion of the numbering of the
variables $\X$.

\section{Pommaret Bases}\label{sec:pombas}

Pommaret bases are a special case of \emph{involutive bases}; see
\cite{wms:comb1} for a general survey.  The algebraic theory of Pommaret bases
was developed in \cite{wms:comb2} (see also \cite[Chpts.~3-5]{wms:invol}).
Given an exponent vector $\mu=[\mu_{1},\dots,\mu_{n}]\neq0$ (or the term
$x^{\mu}$ or a polynomial $f\in\P$ with $\lt{f}=x^{\mu}$ for some fixed term
order), we call $\min{\{i\mid\mu_{i}\neq0\}}$ the \emph{class} of $\mu$ (or
$x^{\mu}$ or $f$), denoted by $\cls{\mu}$ (or $\cls{x^{\mu}}$ or
$\cls{f}$). Then the \emph{multiplicative variables} of $x^{\mu}$ or $f$ are
$\mult{\X}{P}{x^{\mu}}=\mult{\X}{P}{f}=\{x_{1},\dots,x_{\cls{\mu}}\}$.  We say
that $x^{\mu}$ is an \emph{involutive divisor} of another term $x^{\nu}$, if
$x^{\mu}\mid x^{\nu}$ and $x^{\nu-\mu}\in\kk[x_{1},\dots,x_{\cls{\mu}}]$.
Given a finite set $\F\subset\P$, we write $\deg{\F}$ for the maximal degree
and $\cls{\F}$ for the minimal class of an element of $\F$.

\begin{definition}\label{def:pombas}
  Assume first that the finite set $\H\subset\P$ consists only of terms.  $\H$
  is a \emph{Pommaret basis} of the monomial ideal $\I=\langle\H\rangle$, if
  as a $\kk$-linear space
  \begin{equation}\label{eq:pombas}
    \bigoplus_{h\in\H}\kk[\mult{\X}{P}{h}]\cdot h=\I
  \end{equation}
  (in this case each term $x^{\nu}\in\I$ has a unique involutive divisor
  $x^{\mu}\in\H$).  A finite polynomial set $\H$ is a \emph{Pommaret basis} of
  the polynomial ideal $\I$ for the term order $\prec$, if all elements of
  $\H$ possess distinct leading terms and these terms form a Pommaret basis of
  the leading ideal $\lt{\I}$.
\end{definition}

Pommaret bases can be characterised similarly to Gr\"obner bases.  However,
involutive standard representations are unique.  Furthermore, the existence of
a Pommaret basis implies a number of properties that usually hold only
generically.

\begin{proposition}[{\cite[Thm.~5.4]{wms:comb1}}]\label{prop:pbisr}
  The finite set $\H\subset\I$ is a Pommaret basis of the ideal $\I\lhd\P$ for
  the term order $\prec$, if and only if every polynomial $0\neq f\in\I$
  possesses a unique involutive standard representation
  $f=\sum_{h\in\H}P_{h}h$ where each non-zero coefficient
  $P_{h}\in\kk[\mult{\X}{P}{h}]$ satisfies $\lt{(P_{h}h)}\preceq\lt{(f)}$.
\end{proposition}

\begin{proposition}[{\cite[Cor.~7.3]{wms:comb1}}]\label{prop:pbcrit}
  Let $\H$ be a finite set of polynomials and $\prec$ a term order such that
  no leading term in $\lt{\H}$ is an involutive divisor of another one.  The
  set $\H$ is a Pommaret basis of the ideal $\langle\H\rangle$ with respect to
  $\prec$, if and only if for every $h\in\H$ and every non-multiplicative
  index $\cls{h}<j\leq n$ the product $x_{j}h$ possesses an involutive
  standard representation with respect to $\H$.
\end{proposition}

\begin{theorem}[{\cite[Cor.~3.18, Prop.~3.19, Prop.~4.1]{wms:comb2}}]
  \label{thm:pbprop}
   Let $\H$ be a Pommaret basis of the ideal $\I\lhd\P$ for an order $\prec$.
  \begin{description}
  \item[(i)] If $D=\dim{(\P/\I)}$, then $\{x_{1},\dots,x_{D}\}$ is the unique
    maximal strongly independent set modulo $\I$ (and thus
    $\lt{\I}\cap\kk[x_{1},\dots,x_{D}]=\{0\}$).
  \item[(ii)] The restriction of the canonical map $\P\rightarrow\P/\I$ to the
    subring $\kk[x_{1},\dots,x_{D}]$ defines a Noether normalisation.
  \item[(iii)] If $d=\min_{h\in\H}{\cls{h}}$ is the minimal class of a
    generator in $\H$ and $\prec$ is the reverse lexicographic order, then
    $x_{1},\dots,x_{d-1}$ is a maximal $\P/\I$-regular sequence and thus
    $\depth{\P/\I}=d-1$.
  \end{description}
\end{theorem}

The involutive standard representations of the non-multiplicative products
$x_{j}h$ appearing in Proposition \ref{prop:pbcrit} induce a basis of the
first syzygy module.  This observation leads to a stronger version of
Hilbert's syzygy theorem.

\begin{theorem}[{\cite[Thm.~6.1]{wms:comb2}}]\label{thm:pomres}
  Let $\H$ be a Pommaret basis of the ideal\/ $\I\subseteq\P$.  If we denote
  by $\bq{0}{k}$ the number of generators $h\in\H$ with $\cls{\lt{h}}=k$ and
  set $d=\cls{\H}$, then $\I$ possesses a finite free resolution
  \begin{equation}\label{eq:pomres}
    0\longrightarrow\P^{r_{n-d}}\longrightarrow\cdots\longrightarrow
    \P^{r_1}\longrightarrow\P^{r_0}\longrightarrow\I\longrightarrow0
  \end{equation}
  of length\/ $n-d$ where the ranks of the free modules are given by
  \begin{equation}\label{eq:resrank}
    r_i=\sum_{k=d}^{n-i}\binom{n-k}{i}\bq{0}{k}\;.
  \end{equation}
\end{theorem}

We denote by $\reg{\I}$ the \emph{Castelnuovo-Mumford regularity} of $\I$
(considered as a graded module) and by $\pd{\I}$ its \emph{projective
  dimension}.  The \emph{satiety} $\sat{\I}$ is the lowest degree from which
on the ideal $\I$ and its \emph{saturation} $\Isat=\I:\mf^{\infty}$ coincide.
These objects can be easily read off from a Pommaret basis for
$\prec_{\mathrm{revlex}}$.

\begin{theorem}[{\cite[Thm.~8.11,Thm.~9.2, Prop.~10.1, Cor.~10.2]{wms:comb2}}]
  \label{thm:regsatpom}
  Let $\H$ be a Pommaret basis of the ideal $\I\lhd\P$ for the order
  $\prec_{\mathrm{revlex}}$.  We denote by $\H_{1}=\{h\in\H\mid\cls{h}=1\}$
  the subset of generators of class $1$. 
  \begin{description}
  \item[(i)] $\reg{\I}=\deg{\H}$.
  \item[(ii)] $\pd{\I}=n-\cls{\H}$.
  \item[(iii)] Let $\tilde\H_{1}=\{h/x_{1}^{\deg_{x_{1}}{\lt{h}}}\mid
    h\in\H_{1}\}$. Then the set $\bar\H=(\H\setminus\H_{1})\cup\tilde\H_{1}$
    is a weak Pommaret basis\footnote{Thus elimination of redundant generators
      yields a Pommaret basis \cite[Prop.~5.7]{wms:comb1}.} of the saturation
    $\Isat$.  Thus $\Isat=\I:x_{1}^{\infty}$ and the ideal
    $\I$ is saturated, if and only if $\H_{1}=\emptyset$.
  \item[(iv)] $\sat{\I}=\deg{\H_{1}}$.
  \end{description}
\end{theorem}

\begin{remark}\label{rem:extbetti}
  Bayer et al.~\cite{bcp:betti} call a non-vanishing Betti number $\beta_{ij}$
  \emph{extremal}, if $\beta_{k\ell}=0$ for all $k\geq i$ and $\ell>j$.  In
  \cite[Rem.~9.7]{wms:comb2} it is shown how the positions and the values of
  all extremal Betti numbers can be obtained from the Pommaret basis $\H$ for
  $\prec_{\mathrm{revlex}}$.  Let $h_{\gamma_{1}}\in\H$ be of minimal class
  among all generators of maximal degree in $\H$ and set
  $i_{1}=n-\cls{h_{\gamma_{1}}}$ and $q_{1}=\deg{h_{\gamma_{1}}}$.  Then
  $\beta_{i_{1},q_{1}+i_{1}}$ is an extremal Betti number and its value is
  given by the number of generators of degree $q_{1}$ and class $n-i_{1}$.  If
  $\cls{h_{\gamma_{1}}}=\depth{\I}$, it is the only one.  Otherwise let
  $h_{\gamma_{2}}$ be of minimal class among all generators of maximal degree
  in $\{h\in\H\mid\cls{h}<\cls{h_{\gamma_{1}}}\}$.  Defining $i_{2}$, $q_{2}$
  analogous to above, $\beta_{i_{2},q_{2}+i_{2}}$ is a further extremal Betti
  number and its value is given by the number of generators of degree $q_{2}$
  and class $n-i_{2}$ and so on.
\end{remark}

\section{$\delta$-Regularity and Quasi-Stable Ideals}\label{sec:dreg}

Not every ideal $\I\lhd\P$ possesses a finite Pommaret basis.  One can show
that this is solely a problem of the chosen variables~$\X$; after a suitable
linear change of variables $\tilde\X=A\X$ with a non-singular matrix
$A\in\kk^{n\times n}$ the transformed ideal
$\tilde\I\lhd\tilde\P=\kk[\tilde\X]$ has a finite Pommaret basis (for the same
term order which we consider as being defined on exponent vectors)
\cite[Sect.~2]{wms:comb2}.

\begin{definition}\label{def:dreg}
  The variables $\X$ are \emph{$\delta$-regular} for $\I\lhd\P$ and the
  order~$\prec$, if $\I$ has a finite Pommaret basis for $\prec$.
\end{definition}

In \cite[Sect.~2]{wms:comb2} a method is presented to detect effectively
whether given variables are $\delta$-singular and, if this is the case, to
produce deterministically $\delta$-regular variables.  Furthermore, it is
proven there that generic variables are $\delta$-regular so that one can also
employ probabilistic approaches although these are usually computationally
disadvantageous.

It seems to be rather unknown that Serre implicitly presented already in 1964
a version of $\delta$-regularity.  In a letter appended to \cite{gs:alg}, he
introduced the notion of a \emph{quasi-regular} sequence and related it to
Koszul homology.\footnote{Quasi-regular sequences were rediscovered by
  Schenzel et al.~\cite{stc:vcm} under the name \emph{filter-regular}
  sequences and by Aramova and Herzog \cite{ah:almreg} as \emph{almost
    regular} sequences.}  Let $\V$ be a finite-dimensional vector space, $S\V$
the symmetric algebra over $\V$ and $\M$ a finitely generated graded
$S\V$-module.  A vector $v\in\V$ is called quasi-regular at degree $q$ for
$\M$, if $vm=0$ for an $m\in\M$ implies $m\in\M_{<q}$.  A sequence
$(v_{1},\dots,v_{k})$ of vectors $v_{i}\in\V$ is quasi-regular at degree $q$
for $\M$, if each $v_{i}$ is quasi-regular at degree $q$ for $\M/\langle
v_{1},\dots,v_{i-1}\rangle\M$.

Given a basis $\X$ of $\V$, we can identify $S\V$ with the polynomial ring
$\P=\kk[\X]$.  Then it is shown in \cite[Thm.~5.4]{wms:delta} that the
variables $\X$ are $\delta$-regular for a homogeneous ideal $\I\lhd\P$ and the
reverse lexicographic order, if and only if they form a quasi-regular sequence
for the module $\P/\I$ at degree $\reg{\I}$.

Our first result describes the degrees appearing in the Pommaret basis for the
reverse lexicographic order in an intrinsic manner and generalises
\cite[Lemma~2.3]{nvt:reduct} where only Borel-fixed monomial ideals for
$\ch{\kk}=0$ are considered.

\begin{proposition}\label{prop:qi}
  Let the variables $\X$ be $\delta$-regular for the ideal $\I$ and the
  reverse lexicographic order.  If\/ $\H$ denotes the corresponding Pommaret
  basis and $\H_{i}\subseteq\H$ the subset of generators of class $i$, then
  the integer
  \begin{equation}\label{eq:qi}
    q_{i}=
    \max{\bigl\{\,q\in\NN_{0}\mid
                       (\langle\I,x_{1},\dots,x_{i-1}\rangle:x_{i})_{q}\neq
                       \langle\I,x_{1},\dots,x_{i-1}\rangle_{q}\,\bigl\}}
  \end{equation}
  satisfies $q_{i}=\deg{\H_{i}}-1$ (with the convention that
  $\deg{\emptyset}=\max{\emptyset}=-\infty$).
\end{proposition}

\begin{proof}
  Set $\tilde\P=\kk[x_{i},\dots,x_{n}]$ and
  $\tilde\I=\I|_{x_{1}=\cdots=x_{i-1}=0}\lhd\tilde\P$.  Then it is easy to see
  that $q_{i}=\max{\{q\mid(\tilde\I:x_{i})_{q}\neq\tilde\I_{q}\}}$.
  Furthermore, the variables $x_{i},\dots,x_{n}$ are $\delta$-regular for
  $\tilde\I$ and the reverse lexicographic order---the Pommaret basis of
  $\tilde\I$ is given by $\tilde\H=\bigcup_{k\geq i}\tilde\H_{k}$ with
  $\tilde\H_{k}=\H_{k}|_{x_{1}=\cdots=x_{i-1}=0}$
  (cf.~\cite[Lemma~3.1]{wms:noether}).

  Assume first that $\tilde\H_{i}=\emptyset$.  In this case
  $x_{i}f\in\tilde\I$ implies $f\in\tilde\I$, as one can immediately see from
  the involutive standard representation of $x_{i}f$ with respect to
  $\tilde\H$ (all coefficients must lie in $\langle x_{i}\rangle$).  If
  $\tilde\H_{i}\neq\emptyset$, then we choose a generator $\tilde
  h_{\mathrm{max}}\in\tilde\H_{i}$ of maximal degree.  By the properties of
  $\prec_{\mathrm{revlex}}$, we find $\tilde h_{\mathrm{max}}\in\langle
  x_{i}\rangle$ and hence may write $\tilde h_{\mathrm{max}}=x_{i}\tilde g$.
  By definition of a Pommaret basis, $\tilde g\notin\tilde\I$ and thus
  $q_{i}\geq\deg{\tilde g}=\deg{\H_{i}}-1$.

  Assume now that $q_{i}>\deg{\H_{i}}-1$.  Then there exists a polynomial
  $\tilde f\in\tilde\P\setminus\tilde\I$ with $\deg{\tilde f}=q_{i}$ and
  $x_{i}\tilde f\in\tilde\I$.  Consider the involutive standard representation
  $x_{i}\tilde f=\sum_{\tilde h\in\tilde\H}P_{\tilde h}\tilde h$ with respect
  to $\tilde\H$.  If $\cls{\tilde h}>i$, then we must have $P_{\tilde
    h}\in\langle x_{i}\rangle$.  If $\cls{\tilde h}=i$, then by definition
  $P_{\tilde h}\in\kk[x_{i}]$.  Since $\deg{(x_{i}\tilde
    f)}>\deg{\tilde\H_{i}}$, any non-vanishing coefficient $P_{\tilde h}$ must
  be of positive degree in this case.  Thus we can conclude that all
  non-vanishing coefficients $P_{\tilde h}$ lie in $\langle x_{i}\rangle$.
  But then we may divide the involutive standard representation of
  $x_{i}\tilde f$ by $x_{i}$ and obtain an involutive standard representation
  of $\tilde f$ itself so that $\tilde f\in\tilde\I$ in contradiction to the
  assumptions we made.\qed
\end{proof}

Consider the following invariants related to the local cohomology of $\P/\I$
(with respect to the maximal graded ideal $\frak{m}=\langle
x_{1},\dots,x_{n}\rangle$):
\begin{displaymath}
  \begin{aligned}
    a_{i}(\P/\I)&=\max{\{q\mid H^{i}_{\frak{m}}(\P/\I)_{q}\neq0\}}\,,\qquad 
        &0\leq i\leq\dim{(\P/\I)}\,,\\
    \reg_{t}{(\P/\I)}&=\max{\{a_{i}(\P/\I)+i\mid0\leq i\leq t\}}\,,\qquad
        &0\leq t\leq\dim{(\P/\I)}\,,\\
    a^{*}_{t}(\P/\I)&=\max{\{a_{i}(\P/\I)\mid0\leq i\leq t\}}\,,\qquad
        &0\leq t\leq\dim{(\P/\I)}\,.
  \end{aligned}
\end{displaymath}
Trung \cite[Thm.~2.4]{nvt:reduct} related them for monomial Borel-fixed ideals
and $\ch{\kk}=0$ to the degrees of the minimal generators.  We can now
generalise this result to arbitrary homogeneous polynomial ideals.

\begin{corollary}\label{cor:regat}
  Let the variables $\X$ be $\delta$-regular for the ideal $\I\lhd\P$ and the
  reverse lexicographic order.  Denote again by $\H_{i}$ the subset of the
  Pommaret basis $\H$ of $\I$ consisting of the generators of class $i$ and
  set $q_{i}=\deg{\H_{i}}-1$.  Then
  \begin{displaymath}
    \begin{aligned}
      \reg_{t}{(\P/\I)}&=\max{\{q_{1},q_{2},\dots,q_{t+1}\}}\,,\qquad
          &0\leq t\leq\dim{(\P/\I)}\,,\\
      a^{*}_{t}(\P/\I)&=\max{\{q_{1},q_{2}-1,\dots,q_{t+1}-t\}}\,,\qquad
          &0\leq t\leq\dim{(\P/\I)}\,.
    \end{aligned}
  \end{displaymath}
\end{corollary}

\begin{proof}
  This follows immediately from \cite[Thm.~1.1]{nvt:reduct} and
  Proposition~\ref{prop:qi}.\qed 
\end{proof}

For monomial ideals it is in general useless to transform to $\delta$-regular
variables, as the transformed ideal is no longer monomial.  Hence it is a
special property of a monomial ideal to possess a finite Pommaret basis: such
an ideal is called \emph{quasi-stable}.  The following theorem provides
several purely algebraic characterisations of quasi-stability independent of
Pommaret bases.  It combines ideas and results from \cite[Def.~1.5]{bs:mreg},
\cite[Prop.~3.2/3.6]{bg:scmr}, \cite[Prop.~2.2]{hpv:ext} and
\cite[Prop.~4.4]{wms:comb2}.

\begin{theorem}\label{thm:quastab}
  Let $\I\lhd\P$ be a monomial ideal and $D=\dim{(\P/\I)}$.  Then the
  following statements are equivalent.
  \begin{description}
  \item[(i)] $\I$ is quasi-stable.
  \item[(ii)] The variable $x_1$ is not a zero divisor for $\P/\Isat$ and for
    all $1\leq k<D$ the variable $x_{k+1}$ is not a zero divisor for
    $\P/\langle\I,x_1,\dots,x_k\rangle^{\mathrm{sat}}$.
  \item[(iii)] We have $\I:x_1^\infty\subseteq\I:x_2^\infty\subseteq\cdots
    \subseteq\I:x_D^\infty$ and for all $D<k\leq n$ an exponent $e_k\geq1$
    exists such that $x_k^{e_k}\in\I$.
  \item[(iv)] For all $1\leq k\leq n$ the equality $\I:x_k^\infty=\I:\langle
    x_k,\dots,x_n\rangle^\infty$ holds.
  \item[(v)] For every associated prime ideal\/
    $\mathfrak{p}\in\mathrm{Ass}(\P/\I)$ an integer $1\leq j\leq n$ exists
    such that $\mathfrak{p}=\langle x_{j},\dots,x_{n}\rangle$.
  \item[(vi)] If $x^\mu\in\I$ and $\mu_i>0$ for some $1\leq i<n$, then for
    each $0<r\leq\mu_i$ and $i<j\leq n$ an integer $s\geq0$ exists such that\/
    $x_{j}^{s}x^{\mu}/x_{i}^{r}\in\I$.
  \end{description}
\end{theorem}

The terminology ``quasi-stable'' stems from a result of Mall.  The minimality
assumption is essential here, as the simple example $\langle
x^{2},y^{2}\rangle\lhd\kk[x,y]$ shows.

\begin{lemma}[{\cite[Lemma~2.13]{mall:pom},
               \cite[Prop.~5.5.6]{wms:invol}}]\label{lem:stab}   
  A monomial ideal is stable,\footnote{In our ``reverse'' conventions, a
    monomial ideal $\I$ is called \emph{stable}, if for every term $t\in\I$
    and every index $k=\cls{t}<i\leq n$ also $x_{i}t/x_{k}\in\I$.} if and only
  if its minimal basis is a Pommaret basis.
\end{lemma}

Thus already in the monomial case Pommaret bases are generally not minimal.
The following result of Mall characterises those polynomial ideals for which
the reduced Gr\"obner basis is simultaneously a Pommaret basis.  We provide
here a much simpler proof due to a more suitable definition of Pommaret bases.

\begin{theorem}[{\cite[Thm.~2.15]{mall:pom}}]\label{thm:stab}
  The reduced Gr\"obner basis of the ideal\/ $\I\lhd\P$ is simultaneously a
  Pommaret basis, if and only if\/ $\lt{\I}$ is stable.
\end{theorem}

\begin{proof}
  By definition, the leading terms $\lt{\G}$ of a reduced Gr\"obner basis $\G$
  form the minimal basis of $\lt{\I}$.  The assertion is now a trivial
  corollary to Lemma \ref{lem:stab} and the definition of a Pommaret
  basis.\qed
\end{proof}

\section{The Generic Initial Ideal}\label{sec:gin}

If we fix an order $\prec$ and perform a linear change of variables
$\tilde\X=A\X$ with a non-singular matrix $A\in\kk^{n\times n}$, then,
according to Galligo's Theorem \cite{gall:weier,bs:reverse}, for almost all
matrices $A$ the transformed ideal $\tilde\I\lhd\tilde\P=\kk[\tilde\X]$ has
the same leading ideal, the \emph{generic initial ideal} $\gin{\I}$ for the
used order.  By a further result of Galligo \cite{ag:divstab,bs:reverse},
$\gin{\I}$ is Borel fixed, i.\,e.\ invariant under the natural action of the
Borel group.  For $\ch{\kk}=0$, the Borel fixed ideals are precisely the
stable ones; in positive characteristics the property of being Borel fixed has
no longer such a simple combinatorial interpretation.

We will show in this section that many properties of the generic initial ideal
$\gin{\I}$ also hold for the ordinary leading ideal $\lt{\I}$---provided the
used variables are $\delta$-regular.  This observation has a number of
consequences.  While there does not exist an effective criterion for deciding
whether a given leading ideal is actually $\gin{\I}$, $\delta$-regularity is
simply proven by the existence of a finite Pommaret basis.  Furthermore,
$\gin{\I}$ can essentially be computed only by applying a random change of
variables which has many disadvantages from a computational point of view.  By
contrast, \cite[Sect.~2]{wms:comb2} presents a deterministic approach for the
construction of $\delta$-regular variables which in many case will lead to
fairly sparse transformations.  

>From a theoretical point of view, the following trivial lemma which already
appeared in \cite{bs:reverse,gall:weier} implies that proving a statement
about quasi-stable leading ideals immediately entails the analogous statement
about $\gin{\I}$.

\begin{lemma}\label{lem:genqs}
  The generic initial ideal $\gin{\I}$ is quasi-stable.
\end{lemma}

\begin{proof}
  For $\ch{\kk}=0$, the assertion is trivial, since then $\gin{\I}$ is even
  stable, as mentioned above.  For arbitrary $\ch{\kk}$, it follows simply
  from the fact that generic variables\footnote{Recall that we assume
    throughout that $\kk$ is an infinite field, although a sufficiently large
    finite field would also suffice \cite[Rem.~4.3.19]{wms:invol}.} are
  $\delta$-regular and thus yield a quasi-stable leading ideal.\qed
\end{proof}

The next corollary is a classical result \cite[Cor.~1.33]{mlg:gin} for which
we provide here a simple alternative proof.  The subsequent theorem extends
many well-known statements about $\gin{\I}$ to the leading ideal in
$\delta$-regular variables (for $\prec_{\mathrm{revlex}}$); they are all
trivial consequences of the properties of a Pommaret basis.

\begin{corollary}\label{cor:betti}
  Let $\I\lhd\P$ be an ideal and $\ch{\kk}=0$.  Then all bigraded Betti
  numbers satisfy the inequality
  $\beta_{i,j}(\P/\I)\leq\beta_{i,j}(\P/\gin{\I})$.
\end{corollary}

\begin{proof}
  We choose variables $\X$ such that $\lt{\I}=\gin{\I}$.  By Lemma
  \ref{lem:genqs}, these variables are $\delta$-regular for the given ideal
  $\I$.  As $\ch{\kk}=0$, the generic initial ideal is stable and hence the
  bigraded version of (\ref{eq:resrank}) applied to $\lt{\I}$ yields the
  bigraded Betti number $\beta_{i,j}(\P/\gin{\I})$.  Now the claim follows
  immediately from analysing the resolution (\ref{eq:pomres}) degree by
  degree.\qed
\end{proof}

\begin{theorem}\label{thm:deltalt}
  Let the variables $\X$ be $\delta$-regular for the ideal $\I\lhd\P$ and the
  reverse lexicographic order $\prec_{\mathrm{revlex}}$.
  \begin{description}
  \item[(i)] $\pd{\I}=\pd{\lt{\I}}$.
  \item[(ii)] $\sat{\I}=\sat{\lt{\I}}$.
  \item[(iii)] $\reg{\I}=\reg{\lt{\I}}$.
  \item[(iv)] $\reg_{t}{\I}=\reg_{t}{\lt{\I}}$ for all\/ $0\leq
    t\leq\dim{(\P/\I)}$.
  \item[(v)] $a^{*}_{t}(\I)=a^{*}_{t}(\lt{\I})$ for all\/ $0\leq
    t\leq\dim{(\P/\I)}$.
  \item[(vi)] The extremal Betti numbers of $\I$ and $\lt{\I}$ occur at
    the same positions and have the same values.
  \item[(vii)] $\depth{\I}=\depth{\lt{\I}}$.
  \item[(viii)] $\P/\I$ is Cohen-Macaulay, if and only if $\P/\lt{\I}$ is
    Cohen-Macaulay.
  \end{description}
\end{theorem}

\begin{proof}
  The assertions (i-v) are trivial corollaries of Theorem \ref{thm:regsatpom}
  and Corollary \ref{cor:regat}, respectively, where it is shown for all
  considered quantities that they depend only on the leading terms of the
  Pommaret basis of $\I$.  Assertion (vi) is a consequence of Remark
  \ref{rem:extbetti} and the assertions (vii) and (viii) follow from Theorem
  \ref{thm:pbprop}.\qed
\end{proof}

\begin{remark}
  In view of Part (viii), one may wonder whether a similar statement holds for
  Gorenstein rings.  In \cite[Ex.~5.5]{wms:noether} the ideal $\I=\langle
  z^{2}-xy, yz, y^{2}, xz, x^{2}\rangle\lhd\kk[x,y,z]$ is studied.  The used
  coordinates are $\delta$-regular for $\prec_{\mathrm{revlex}}$, as a
  Pommaret basis is obtained by adding the generator $x^{2}y$.  It follows
  from \cite[Thm.~5.4]{wms:noether} that $\P/\I$ is Gorenstein, but
  $\P/\lt{\I}$ not.  A computation with \cocoa\ \cite{cocoa} gives here
  $\gin{\I}=\langle z^{2}, yz, y^{2}, xz, xy, x^{3}\rangle$ (assuming
  $\ch{\kk}=0$) and again one may conclude with \cite[Thm.~5.4]{wms:noether}
  that $\P/\gin{\I}$ is not Gorenstein.
\end{remark}

\section{Componentwise Linear Ideals}\label{sec:complinid}

Given an ideal $\I\lhd\P$, we denote by $\I_{\langle d\rangle}=
\langle\I_{d}\rangle$ the ideal generated by the homogeneous
component~$\I_{d}$ of degree $d$.  Herzog and Hibi \cite{hh:complin} called
$\I$ \emph{componentwise linear}, if for every degree $d\geq0$ the ideal
$\I_{\langle d\rangle}= \langle\I_{d}\rangle$ has a linear resolution.  For a
connection with Pommaret bases, we need a refinement of $\delta$-regularity.

\begin{definition}
  The variables $\X$ are \emph{componentwise $\delta$-regular} for the ideal
  $\I$ and the order $\prec$, if all ideals $\I_{\langle d\rangle}$ for
  $d\geq0$ have finite Pommaret bases for $\prec$.
\end{definition}

It follows from the proof of \cite[Thm.~9.12]{wms:comb2} that for the
definition of componentwise $\delta$-regularity it suffices to consider the
finitely many degrees $d\leq\reg{\I}$.  Thus trivial modificiations of any
method for the construction of $\delta$-regular variables allow to determine
effectively componentwise $\delta$-regular variables.

\begin{theorem}[{\cite[Thm.~8.2, Thm.~9.12]{wms:comb2}}]\label{thm:complinres}
  Let the variables $\X$ be componentwise $\delta$-regular for the ideal
  $\I\lhd\P$ and the reverse lexicographic order.  If\/ $\I$ is componentwise
  linear, then the free resolution (\ref{eq:pomres}) of\/ $\I$ induced by the
  Pommaret basis $\H$ is minimal and the Betti numbers of\/ $\I$ are given by
  (\ref{eq:resrank}).  Conversely, if the resolution (\ref{eq:pomres}) is
  minimal, then the ideal $\I$ is componentwise linear.
\end{theorem}

The following corollary generalises the analogous result for stable ideals to
componentwise linear ideals (Aramova et al. \cite[Thm.~1.2(a)]{ahh:stabbetti}
noted a version for $\gin{\I}$).  It is an immediate consequence of the linear
construction of the resolution (\ref{eq:pomres}) in \cite[Thm.~6.1]{wms:comb2}
and its minimality for componentwise linear ideals.

\begin{corollary}\label{cor:betticl}
  Let $\I\lhd\P$ be componentwise linear.  If the Betti number $\beta_{i,j}$
  does not vanish, then also all Betti numbers $\beta_{i',j}$ with $i'<i$ do
  not vanish.
\end{corollary}

As a further corollary, we obtain a simple proof of an estimate given by
Aramova et al.~\cite[Cor.~1.5]{ahh:stabbetti} (based on
\cite[Thm.~2]{hk:betti}).

\begin{corollary}\label{cor:estbetti}
  Let $\I\lhd\P$ be a componentwise linear ideal with $\pd{\I}=p$.  Then
  the Betti numbers satisfy $\beta_{i}\geq\binom{p+1}{i+1}$.
\end{corollary}

\begin{proof}
  Let $\H$ be the Pommaret basis of $\I$ for $\prec_{\mathrm{revlex}}$ in
  componentwise $\delta$-regular variables and $d=\cls{\H}$.  By Theorem
  \ref{thm:complinres}, (\ref{eq:pomres}) is the minimal resolution of $\I$
  and hence (\ref{eq:resrank}) gives us $\beta_{i}$.  By Theorem
  \ref{thm:pbprop}, $p=n-d$.  We also note that $\delta$-regularity implies
  that $\bq{0}{k}>0$ for all $d\leq k\leq n$.  Now we compute
  \begin{displaymath}
    \beta_{i}=\sum_{k=d}^{n-i}\binom{n-k}{i}\bq{0}{k}=
        \sum_{\ell=i}^{p}\binom{\ell}{i}\bq{0}{n-\ell}\geq
        \sum_{\ell=i}^{p}\binom{\ell}{i}=\binom{p+1}{i+1}
  \end{displaymath}
  by a well-known identity for binomial coefficients. \qed
\end{proof}

\begin{example}
  The estimate in Corollary \ref{cor:estbetti} is sharp.  It is realised by
  any componentwise linear ideal whose Pommaret basis satisfies $\bq{0}{i}=0$
  for $i<d$ and $\bq{0}{i}=1$ for $i\geq d$.  As a simple monomial example
  consider the ideal $\I$ generated by the $d$ terms
  $h_{1}=x_{n}^{\alpha_{n}+1}$,
  $h_{2}=x_{n}^{\alpha_{n}}x_{n-1}^{\alpha_{n-1}+1}$,\dots,
  $h_{d}=x_{n}^{\alpha_{n}}\cdots x_{d+1}^{\alpha_{d+1}}x_{d}^{\alpha_{d}+1}$
  for arbitrary exponents $\alpha_{i}\geq0$.  One easily verifies that
  $\H=\{h_{1},\dots,h_{d}\}$ is indeed simultaneously the Pommaret and the
  minimal basis of $\I$.
\end{example}

Recently, Nagel and R\"omer \cite[Thm.~2.5]{nr:complin} provided some criteria
for componentwise linearity based on $\gin{\I}$ (see also
\cite[Thm~1.1]{ahh:stabbetti} where the case $\ch{\kk}=0$ is treated).  We
will now show that again $\gin{\I}$ may be replaced by $\lt{\I}$, if one uses
componentwise $\delta$-regular variables.  Furthermore, our proof is
considerably simpler than the one by Nagel and R\"omer.

\begin{theorem}\label{thm:complincrit}
  Let the variables $\X$ be componentwise $\delta$-regular for the ideal
  $\I\lhd\P$ and the reverse lexicographic order.  Then the following
  statements are equivalent:
  \begin{description}
  \item[(i)] $\I$ is componentwise linear.
  \item[(ii)] $\lt{\I}$ is stable and all bigraded Betti numbers $\beta_{ij}$
    of\/ $\I$ and\/ $\lt{\I}$ coincide.
  \item[(iii)] $\lt{\I}$ is stable and all total Betti numbers $\beta_{i}$
    of\/ $\I$ and\/ $\lt{\I}$ coincide.
  \item[(iv)] $\lt{\I}$ is stable and $\beta_{0}(\I)=\beta_{0}(\lt{\I})$.
  \end{description}
\end{theorem}

\begin{proof}
  The implication ``(i)${}\Rightarrow{}$(ii)'' is a simple consequence of
  Theorem \ref{thm:complinres}.  Since our variables are componentwise
  $\delta$-regular, the resolution (\ref{eq:pomres}) is minimal.  This implies
  immediately that $\lt{\I}$ is stable.  Applying Theorem \ref{thm:pomres} to
  the Pommaret basis $\lt{\H}$ of $\lt{\I}$ yields the minimal resolution of
  $\lt{\I}$.  In both cases, the leading terms of all syzygies are determined
  by $\lt{\H}$ and hence the bigraded Betti numbers of $\I$ and $\lt{\I}$
  coincide.

  The implications ``(ii)${}\Rightarrow{}$(iii)'' and
  ``(iii)${}\Rightarrow{}$(iv)'' are trivial.  Thus there only remains to
  prove ``(iv)${}\Rightarrow{}$(i)''.  Let $\H$ be the Pommaret basis of $\I$.
  Since $\lt{\I}$ is stable by assumption, $\lt{\H}$ is its minimal basis by
  Lemma \ref{lem:stab} and $\beta_{0}(\lt{\I})$ equals the number of elements
  of $\H$.  The assumption $\beta_{0}(\I)=\beta_{0}(\lt{\I})$ implies that
  $\H$ is a minimal generating system of $\I$.  Hence, none of the syzygies
  obtained from the involutive standard representations of the
  non-multiplicative products $yh$ with $h\in\H$ and $y\in\nmult{\X}{P}{h}$
  may contain a non-vanishing constant coefficients.  By
  \cite[Lemma~8.1]{wms:comb2}, this observation implies that the resolution
  (\ref{eq:pomres}) induced by $\H$ is minimal and hence the ideal $\I$ is
  componentwise linear by Theorem \ref{thm:complinres}.\qed
\end{proof}

\section{Linear Quotients}\label{sec:linquot}

Linear quotients were introduced by Herzog and Takayama \cite{ht:rescones} in
the context of constructing iteratively a free resolution via mapping cones.
As a special case, they considered monomial ideals where certain colon ideals
defined by an ordered minimal basis are generated by variables. Their
definition was generalised by Sharifan and Varabaro \cite{sv:linquot} to
arbitrary ideals.

\begin{definition}
  Let $\I\lhd\P$ be an ideal and $\F=\{f_{1},\dots,f_{r}\}$ an ordered basis
  of it.  Then $\I$ has \emph{linear quotients} with respect to $\F$, if for
  each $1<k\leq r$ the ideal $\langle f_{1},\dots,f_{k-1}\rangle:f_{k}$ is
  generated by a subset $\X_{k}\subseteq\X$ of variables.
\end{definition}

We show first that in the monomial case this concept captures the essence of a
Pommaret basis.  For this purpose, we ``invert'' some notions introduced in
\cite{wms:comb2}.  We associate with a monomial Pommaret basis $\H$ a directed
graph, its \emph{$P$-graph}.  Its vertices are the elements of $\H$.  Given a
non-multiplicative variable $x_j\in\nmult{\X}{P}{h}$ for a generator $h\in\H$,
there exists a unique involutive divisor $\bar h\in\H$ of $x_jh$ and we
include a directed edge from $h$ to $\bar h$.

An ordering of the elements of $\H$ is called an \emph{inverse $P$-ordering},
if $\alpha>\beta$ whenever the $P$-graph contains a path from $h_{\alpha}$ to
$h_{\beta}$.  It is straightforward to describe explicitly an inverse
$P$-ordering: we set $\alpha>\beta$, if $\cls{h_{\alpha}}<\cls{h_{\beta}}$ or
if $\cls{h_{\alpha}}=\cls{h_{\beta}}$ and
$h_{\alpha}\prec_{\mathrm{lex}}h_{\beta}$, i.\,e.\ we sort the generators
$h_\alpha$ first by their class and then within each class lexicographically
(according to our reverse conventions!).  One easily verifies that this
defines an inverse $P$-ordering.

\begin{multicols}{2}
\begin{example}\label{ex:Pgraph}
  Consider the monomial ideal $\I\subset\kk[x,y,z]$ generated by the six terms
  $h_1=z^2$, $h_2=yz$, $h_3=y^2$, $h_4=xz$, $h_5=xy$ and $h_6=x^2$.  One
  easily verifies that these terms form a Pommaret basis of $\I$.  The
  $P$-graph in (\ref{eq:Pgraph}) shows that the generators are already
  inversely $P$-ordered, namely according to the description above.

  
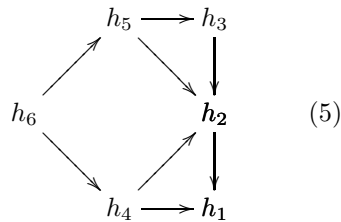
\begin{equation}\label{eq:Pgraph}
    \begin{xy}
      \xygraph{%
        !~:{@{->}}%
        []{h_6}
        (:[ur]{h_5}
         (:[r]{h_3}:[d]{h_2}:[d]{h_1},
          :[dr]{h_2}),
         :[dr]{h_4}
          (:[ur]{h_2},:[r]{h_1}))}
    \end{xy}
  \end{equation}  
\end{example}  
\end{multicols}

\begin{proposition}\label{prop:linquot}
  Let $\H=\{h_{1},\dots,h_{r}\}$ be an inversely $P$-ordered monomial Pommaret
  basis of the quasi-stable monomial ideal $\I\lhd\P$.  Then the ideal $\I$
  possesses linear quotients with respect to the basis $\H$ and
  \begin{equation}\label{eq:colon}
    \langle h_{1},\dots,h_{k-1}\rangle:h_{k}=\langle\nmult{\X}{P}{h_{k}}\rangle
    \qquad k=1,\dots r\;.
  \end{equation}
  Conversely, assume that $\H=\{h_{1},\dots,h_{r}\}$ is a monomial generating
  set of the monomial ideal $\I\lhd\P$ such that (\ref{eq:colon}) is
  satisfied.  Then $\I$ is quasi-stable and $\H$ its Pommaret basis.
\end{proposition}

\begin{proof}
  Let $y\in\nmult{\X}{P}{h_{k}}$ be a non-multiplicative variable for
  $h_{k}\in\H$.  Since $\H$ is a Pommaret basis, the product $yh_{k}$
  possesses an involutive divisor $h_{i}\in\H$ and, by definition, the
  $P$-graph of $\H$ contains an edge from $k$ to $i$.  Thus $i<k$ for an
  inverse $P$-ordering, which proves the inclusion ``$\supseteq$''.

  The following argument shows that the inclusion cannot be strict.  Consider
  a term $t\in\kk[\mult{\X}{P}{h_{k}}]$ consisting entirely of multiplicative
  variables and assume that $th_{k}\in\langle h_{1},\dots,h_{k-1}\rangle$,
  i.\,e.\ $th_{k}=s_{1}h_{i_{1}}$ for some term $s_{1}\in\kk[\X]$ and some
  index $i_{1}<k$.  By definition of a Pommaret basis, $s_{1}$ must contain at
  least one non-multiplicative variable $y_{1}$ of $h_{i_{1}}$.  But now we
  may rewrite $y_{1}h_{i_{1}}=s_{2}h_{i_{2}}$ with $i_{2}<i_{1}$ and
  $s_{2}\in\kk[\mult{\X}{P}{h_{i_{2}}}]$.  Since this implies
  $\cls{h_{2}}\geq\cls{h_{1}}$, we find
  $\mult{\X}{P}{h_{i_{1}}}\subseteq\mult{\X}{P}{h_{i_{2}}}$.  Hence after a
  finite number of iterations we arrive at a representation $th_{k}=sh_{i}$
  where $s\in\kk[\mult{\X}{P}{h_{i}}]$ which is, however, not possible for a
  Pommaret basis.

  For the converse, we show by a finite induction over $k$ that every
  non-multi\-plicative product $yh_{k}$ with $y\in\nmult{\X}{P}{h_{k}}$
  possesses an involutive divisor $h_{i}$ with $i<k$ which implies our
  assertion by Proposition \ref{prop:pbcrit}.  For $k=1$ nothing is to be
  shown, since (\ref{eq:colon}) implies in this case that all variables are
  multiplicative for $h_{1}$ (and thus this generator is of the form
  $h_{1}=x_{n}^{\ell}$ for some $\ell>0$), and $k=2$ is trivial.  Assume that
  our claim was true for $h_{1},h_{2},\dots,h_{k-1}$.  Because of
  (\ref{eq:colon}), we may write $yh_{k}=t_{1}h_{i_{1}}$ for some $i_{1}<k$.
  If $t_{1}\in\kk[\mult{\X}{P}{h_{i_{1}}}]$, we set $i=i_{1}$ and are done.
  Otherwise, $t_{1}$ contains a non-multiplicative variable
  $y_{1}\in\nmult{\X}{P}{h_{i_{1}}}$.  By our induction assumption,
  $y_{1}h_{i_{1}}$ has an involutive divisor $h_{i_{2}}$ with $i_{2}<i_{1}$
  leading to an alternative representation $yh_{k}=t_{2}h_{i_{2}}$.  Now we
  iterate and find after finitely many steps an involutive divisor $h_{i}$ of
  $yh_{k}$, since the sequence $i_{1}>i_{2}>\cdots$ is strictly decreasing and
  $h_{1}$ has no non-multiplicative variables.\qed
\end{proof}

\begin{remark}
  As we are here exclusively concerned with Pommaret bases, we formulated and
  proved the above result only for this special case.  However, Proposition
  \ref{prop:linquot} remains valid for any involutive basis with respect to a
  \emph{continuous} involutive division $L$ (and thus for all divisions of
  practical interest).  The continuity of $L$ is needed here for two reasons.
  Firstly, it guarantees the existence of an $L$-ordering, as for such
  divisions the $L$-graph is always acyclic \cite[Lemma~5.5]{wms:comb2}.
  Secondly, the above argument that finitely many iterations lead to a
  representation $th_{k}=sh_{i}$ where $s$ contains only multiplicative
  variables for $h_{i}$ is specific for the Pommaret division and cannot be
  generalised.  However, the very definition of continuity
  \cite[Def.~4.9]{gb:invbas} ensures that for continuous divisions such a
  rewriting cannot be done infinitely often.
\end{remark}

In general, we cannot expect that the second part of Proposition
\ref{prop:linquot} remains true, when we consider arbitrary polynomial ideals.
However, for the first part we find the following variation of
\cite[Thm.~2.3]{sv:linquot}.

\begin{proposition}\label{prop:linquot2}
  Let $\H$ be a Pommaret basis of the polynomial ideal $\I\lhd\P$ for the term
  order $\prec$ and $h'\in\P$ a polynomial with $\lt{h'}\notin\lt{\H}$.  If
  $\I:h'= \langle\nmult{\X}{P}{h'}\rangle$, then $\H'=\H\cup\{h'\}$ is a
  Pommaret basis of $\J=\I+\langle h'\rangle$.  If furthermore $\I$ is
  componentwise linear, the variables $\X$ are componentwise $\delta$-regular
  and $\H'$ is a minimal basis of $\J$, then $\J$ is componentwise linear,
  too.
\end{proposition}

\begin{proof}
  If $\I:h'= \langle\nmult{\X}{P}{h'}\rangle$, then all products of $h'$ with
  one of its non-multiplicative variables lie in $\I$ and hence possess an
  involutive standard representation with respect to $\H$.  This immediately
  implies the first assertion.

  In componentwise $\delta$-regular variables all syzygies obtained from the
  involutive standard representations of products $yh$ with $h\in\H$ and
  $y\in\nmult{\X}{P}{h}$ are free of constant coefficients, if $\I$ is
  componentwise linear.  If $\H'$ is a minimal basis of $\J$, the same is true
  for all syzygies obtained from products $yh'$ with $y\in\nmult{\X}{P}{h'}$.
  Hence we can again conclude with \cite[Lemma~8.1]{wms:comb2} that the
  resolution of $\J$ induced by $\H'$ is minimal and $\J$ componentwise linear
  by Theorem \ref{thm:complinres}.\qed
\end{proof}


\bibliography{QuasiGen.bib}

\begin{thebibliography}{10}
\providecommand{\url}[1]{\texttt{#1}}
\providecommand{\urlprefix}{URL }

\bibitem{ah:almreg}
Aramova, A., Herzog, J.: Almost regular sequences and {B}etti numbers. Amer. J.
  Math.  122,  689--719 (2000)

\bibitem{ahh:stabbetti}
Aramova, A., Herzog, J., Hibi, T.: Ideals with stable {B}etti numbers. Adv.\
  Math.  152,  72--77 (2000)

\bibitem{bcp:betti}
Bayer, D., Charalambous, H., Popescu, S.: Extremal {B}etti numbers and
  applications to monomial ideals. J.\ Alg.  221,  497--512 (1999)

\bibitem{bs:mreg}
Bayer, D., Stillman, M.: A criterion for detecting $m$-regularity. Invent.\
  Math.  87,  1--11 (1987)

\bibitem{bs:reverse}
Bayer, D., Stillman, M.: A theorem on refining division orders by the reverse
  lexicographic orders. Duke J.\ Math.  55,  321--328 (1987)

\bibitem{bg:scmr}
Bermejo, I., Gimenez, P.: Saturation and {C}astelnuovo-{M}umford regularity.
  J.\ Alg.  303,  592--617 (2006)

\bibitem{cs:reg}
Caviglia, G., Sbarra, E.: Characteristic-free bounds for the
  {C}astelnuovo-{M}umford regularity. Compos.\ Math.  141,  1365--1373 (2005)

\bibitem{cocoa}
{CoCoA}Team: \cocoa: a system for doing {C}omputations in {C}ommutative
  {A}lgebra. Available at \/ {\tt http://cocoa.dima.unige.it}

\bibitem{de:ca}
Eisenbud, D.: Commutative Algebra with a View Toward Algebraic Geometry.
  Graduate Texts in Mathematics 150, Springer-Verlag, New York (1995)

\bibitem{gall:weier}
Galligo, A.: A propos du th\'eor\`eme de pr\'eparation de {W}eierstrass. In:
  Norguet, F. (ed.) Fonctions de Plusieurs Variables Complexes, pp. 543--579.
  Lecture Notes in Mathematics 409, Springer-Verlag, Berlin (1974)

\bibitem{ag:divstab}
Galligo, A.: Th\'eor\`eme de division et stabilit\'e en g\'eometrie analytique
  locale. Ann.\ Inst.\ Fourier  29(2),  107--184 (1979)

\bibitem{gb:invbas}
Gerdt, V., Blinkov, Y.: Involutive bases of polynomial ideals. Math.\ Comp.\
  Simul.  45,  519--542 (1998)

\bibitem{mlg:gin}
Green, M.: Generic initial ideals. In: Elias, J., Giral, J., Mir{\'o}-Roig, R.,
  Zarzuela, S. (eds.) Six Lectures on Commutative Algebra, pp. 119--186.
  Progress in Mathematics~166, Birkh\"auser, Basel (1998)

\bibitem{gs:alg}
Guillemin, V., Sternberg, S.: An algebraic model of transitive differential
  geometry. Bull.\ Amer.\ Math.\ Soc.  70,  16--47 (1964), (With a letter of
  Serre as appendix)

\bibitem{wms:delta}
Hausdorf, M., Sahbi, M., Seiler, W.: $\delta$- and quasi-regularity for
  polynomial ideals. In: Calmet, J., Seiler, W., Tucker, R. (eds.) Global
  Integrability of Field Theories, pp. 179--200. Universit\"atsverlag
  Karlsruhe, Karls\-ruhe (2006)

\bibitem{hh:complin}
Herzog, J., Hibi, T.: Componentwise linear ideals. Nagoya Math.\ J.  153,
  141--153 (1999)

\bibitem{hh:monid}
Herzog, J., Hibi, T.: Monomial Ideals. Graduate Texts in Mathematics~260,
  Springer-Verlag, London (2011)

\bibitem{hk:betti}
Herzog, J., K{\"u}hl, M.: On the {B}ettinumbers of finite pure and linear
  resolutions. Comm.\ Alg.  12,  1627--1646 (1984)

\bibitem{hpv:ext}
Herzog, J., Popescu, D., Vladoiu, M.: On the {E}xt-modules of ideals of {B}orel
  type. In: Commutative Algebra, pp. 171--186. Contemp.\ Math.~331, Amer.\
  Math.\ Soc., Providence (2003)

\bibitem{ht:rescones}
Herzog, J., Takayama, Y.: Resolutions by mapping cones. Homol.\ Homot.\ Appl.
  4,  277--294 (2002)

\bibitem{mall:pom}
Mall, D.: On the relation between {Gr\"obner} and {P}ommaret bases. Appl.\
  Alg.\ Eng.\ Comm.\ Comp.  9,  117--123 (1998)

\bibitem{nr:complin}
Nagel, U., R{\"o}mer, T.: Criteria for componentwise linearity. Preprint
  arXiv:1108.3921 (2011)

\bibitem{stc:vcm}
Schenzel, P., Trung, N., Cuong, N.: Verallgemeinerte
  {C}ohen-{M}acaulay-{M}oduln. Math.\ Nachr.  85,  57--73 (1978)

\bibitem{wms:comb1}
Seiler, W.: A combinatorial approach to involution and $\delta$-regularity {I}:
  Involutive bases in polynomial algebras of solvable type. Appl.\ Alg.\ Eng.\
  Comm.\ Comp.  20,  207--259 (2009)

\bibitem{wms:comb2}
Seiler, W.: A combinatorial approach to involution and $\delta$-regularity
  {II}: Structure analysis of polynomial modules with {P}ommaret bases. Appl.\
  Alg.\ Eng.\ Comm.\ Comp.  20,  261--338 (2009)

\bibitem{wms:invol}
Seiler, W.: Involution --- {T}he {F}ormal {T}heory of {D}ifferential
  {E}quations and its {A}pplications in {C}omputer {A}lgebra. Algorithms and
  Computation in Mathematics~24, Springer-Verlag, Berlin (2009)

\bibitem{wms:noether}
Seiler, W.: Effective genericity, $\delta$-regularity and strong {N}oether
  position. Comm.\ Alg.  (to appear)

\bibitem{sv:linquot}
Sharifan, L., Varbaro, M.: Graded {B}etti numbers of ideals with linear
  quotients. Matematiche  63,  257--265 (2008)

\bibitem{nvt:reduct}
Trung, N.: Gr\"obner bases, local cohomology and reduction number. Proc.\
  Amer.\ Math.\ Soc.  129,  9--18 (2001)

\end{thebibliography}
\bibliographystyle{splncs03}

\end{document}